\documentclass[letterpaper, 10 pt, conference]{ieeeconf}                                                             
\IEEEoverridecommandlockouts                             
\usepackage{graphics} 
\usepackage{epsfig} 
\usepackage{mathptmx} 
\usepackage{amsmath} 
\usepackage{amssymb}  
\usepackage{graphicx}
\usepackage{subcaption}
\usepackage{color,colortbl}
\usepackage{cleveref}
\usepackage{url}
\usepackage{enumerate}

\newtheorem{theorem}{Theorem}

\newtheorem{problem}[theorem]{Problem}

\newtheorem{remark}{Remark}

\newtheorem{assumption}{Assumption}

\title{\LARGE \bf A simple controller for the transition maneuver of a tail-sitter drone}

\author{A. Flores, A. Montes de Oca and G. Flores \thanks{A. Flores, A. Montes de Oca and G. Flores are with the Perception and Robotics Laboratory, Centro de Investigaciones en \'{O}ptica, Le\'{o}n, Guanajuato, Mexico, 37150. (email: alejandrofl@cio.mx, andresmr@cio.mx, gflores@cio.mx). Corresponding author: Gerardo Flores.} 
\thanks{This work was supported by the FORDECYT-CONACYT under grant 000000000292399.}
}

\begin{document}
\maketitle
\thispagestyle{empty}
\pagestyle{empty}
\begin{abstract}
This paper presents a controller for the transition maneuver of a tail-sitter drone. The tail-sitter model considers aerodynamic terms whereas the proposed controller considers the time-scale separation between drone attitude and position dynamics. The controller design is based on Lyapunov approach and linear saturation functions. Simulations experiments demonstrate the effectiveness of the derived theoretical results.
\end{abstract}
\IEEEpeerreviewmaketitle
\section{Introduction} \label{sec:intro}
UAVs designs can be divided into two principal types: fixed-wing and the rotatory aircraft, both have their respective pros and cons. Nonetheless, it is possible to obtain certain capabilities of both of them by using the so-called convertible aircraft. There exist two main configurations of convertible aircraft: tilt-rotors \cite{1,7040348, GFlores:thesis, flores_cuav_jirs, Flores:IFAC14_convertible} and tail-sitter \cite{4,9,10}. The tail-sitter transition maneuver is considerable more complex than tilt-rotors, since it is required that the complete UAV body tilt around one axis for achieving the \textit{transition maneuver}, an ordered movement that carries the UAV from hover to cruise mode and vice versa, see Fig. \ref{fig:cambio}. This paper is focused to develop and implement a controller especially dedicated to achieve such transition maneuver.

Most of the papers related with this topic tackle the problem from a simplified model without taking into account aerodynamics, bounded limits of actuators, or even smoothness of the controller. In \cite{1,5,14} the control strategy consists in increasing the motor angles until they are horizontally aligned, then a switching strategy between flight modes is performed. In \cite{19} authors develop a robust algorithm to control airplane and hover flight modes considering aerodynamic terms, however it does not present a transition control. In \cite{16} the transition is performed manually and the control relies in modifying the flight modes control weights according to the percentage of the transition progress. In \cite{15,8444168,8327537,LI2018348} it is implemented a controller in for the transition maneuver of a tail-sitter UAV, the controller is designed in the 6-DOF, on the other hand, however \cite{8431535} lacks of a theoretical proof which demonstrate the effectiveness of the transition maneuver. Also, several unified controllers which avoids control switching methods are investigated in \cite{8206359,8463158}.
\begin{figure}[ht!]
	\centering
	\includegraphics[width=250pt,page=1]{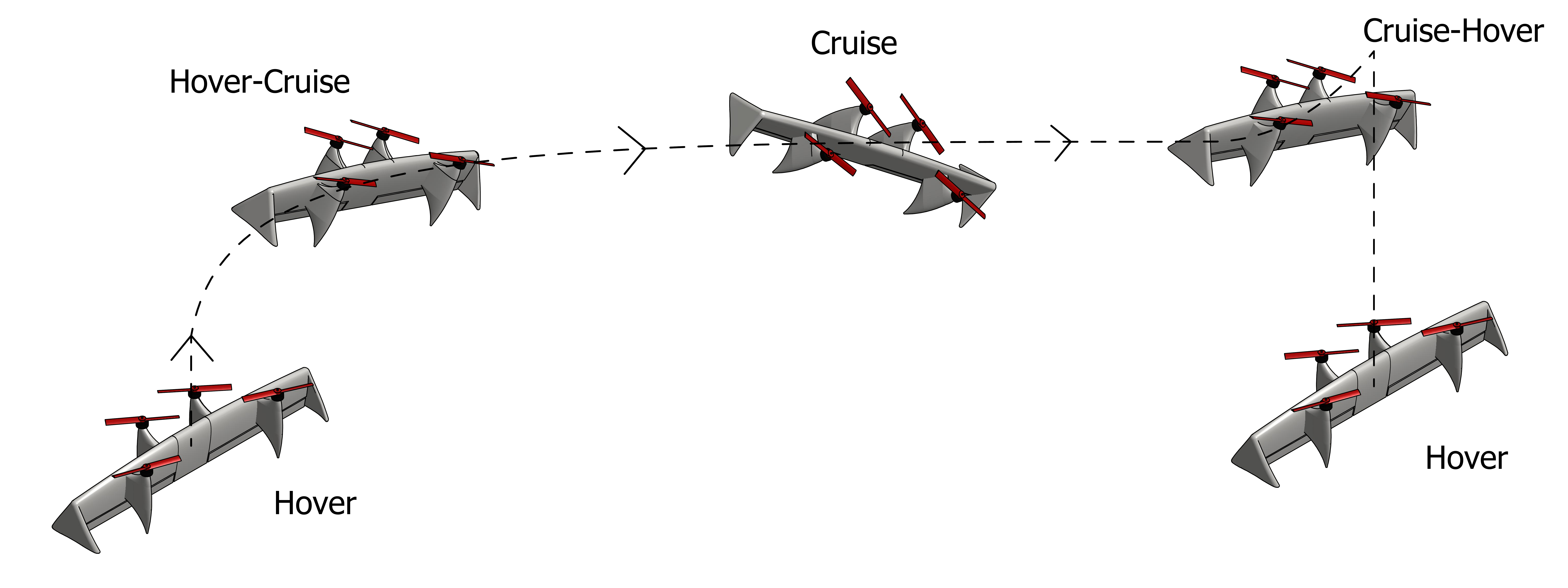}
	\caption{Transition maneuver, graphical description.}
	\label{fig:cambio}
\end{figure}

In this paper we present a mathematical UAV model which considers aerodynamic effects. The proposed controller takes into account the usual available states in a real scenario. Also, the control design avoids switching dynamics between modes, having a sweet and bounded controllers, which are desired properties in real implementations. The controller designed is based on the time-scale separation property presented in attitude and position UAV dynamics \cite{20,21}. The control design is used saturation functions and Lyapunov stability.

The remainder of this paper is organized as follows. In Section \ref{sec:problem} system model and problem statement are presented. Section \ref{sec:result} presents the control system proposed and the followed strategy to ensure a successful flight mode transition. In Section \ref{sec:sims} simulations experiments demonstrate the effectiveness of the theoretical results. Finally, some conclusions and further work are described in Section \ref{sec:conclusions}.
%
\section{Problem Formulation} \label{sec:problem}
The tail-sitter UAV have three operation modes: a) hover; b) cruise; and c) transition mode. The last mode consists in the transition phase between cruise (hover) to hover (cruise); this flying mode will be investigated in this paper. In hover to cruise mode the attitude changes from $90$ to approximately $6$ degrees around between the $y$ axis, which is the range for cruise flight mode. It is the same transition from cruise to hover mode but with the angle reversed.
%
\begin{figure}[ht!]
	\centering
	\includegraphics[width=180pt,page=2]{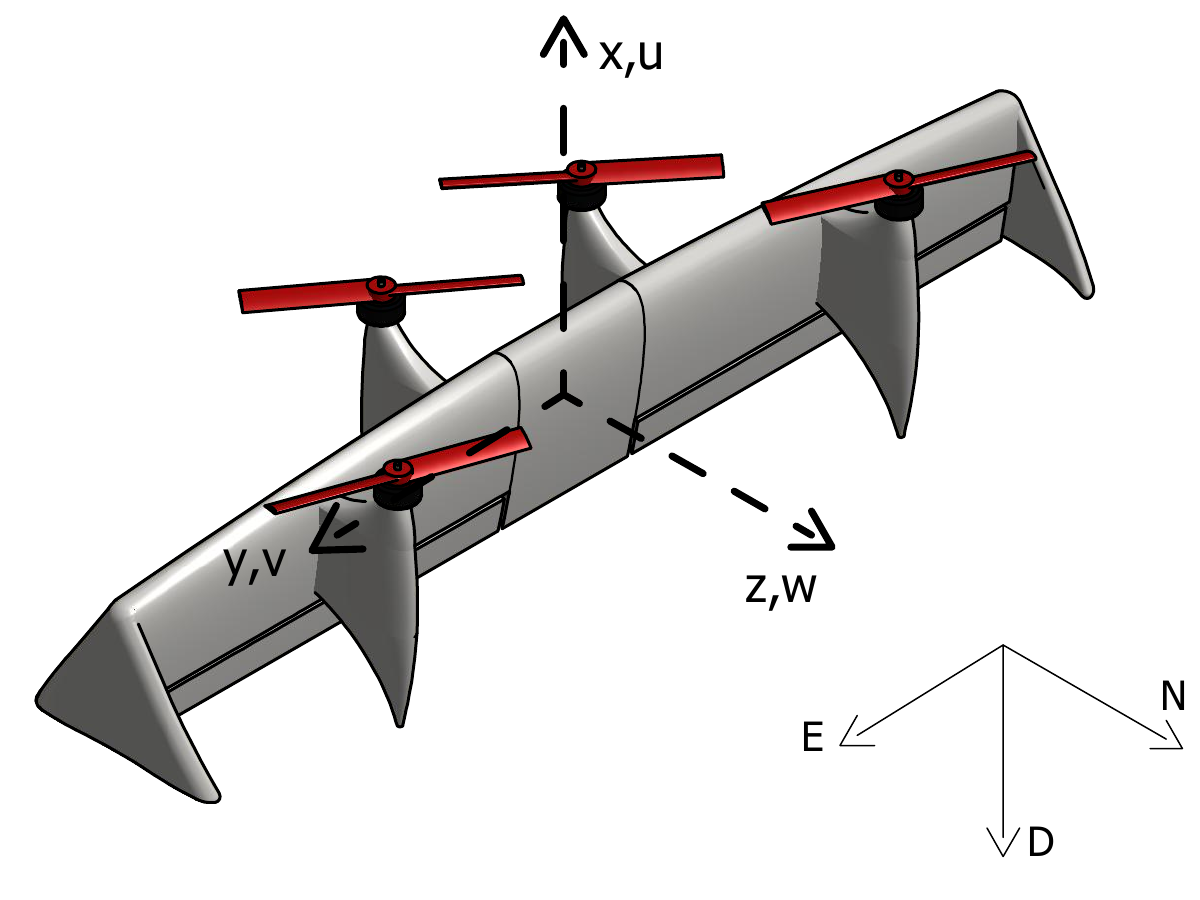}
	\caption{Tail-sitter aircraft; body and world coordinate frames.}
	\label{frames}
\end{figure}
The drone is composed by six actuators: four rotors and two servos which manipulate attitude dynamics. These six actuators results in an over-actuated roll and pitch subsystems. Figs. \ref{fig:cambio} and \ref{frames} depict the tail-sitter drone used in this study.

\subsection{System Model}
The analysis and control design is implemented on the $(x,z)$ plane from the body frame. This is the plane where the most interesting dynamics for the transition maneuver appears \cite{6,8}. Forces affecting aircraft behavior are shown in Fig. \ref{forces}. The system can be described as follows
%
\begin{figure}[ht!]
	\centering
	\includegraphics[width=180pt]{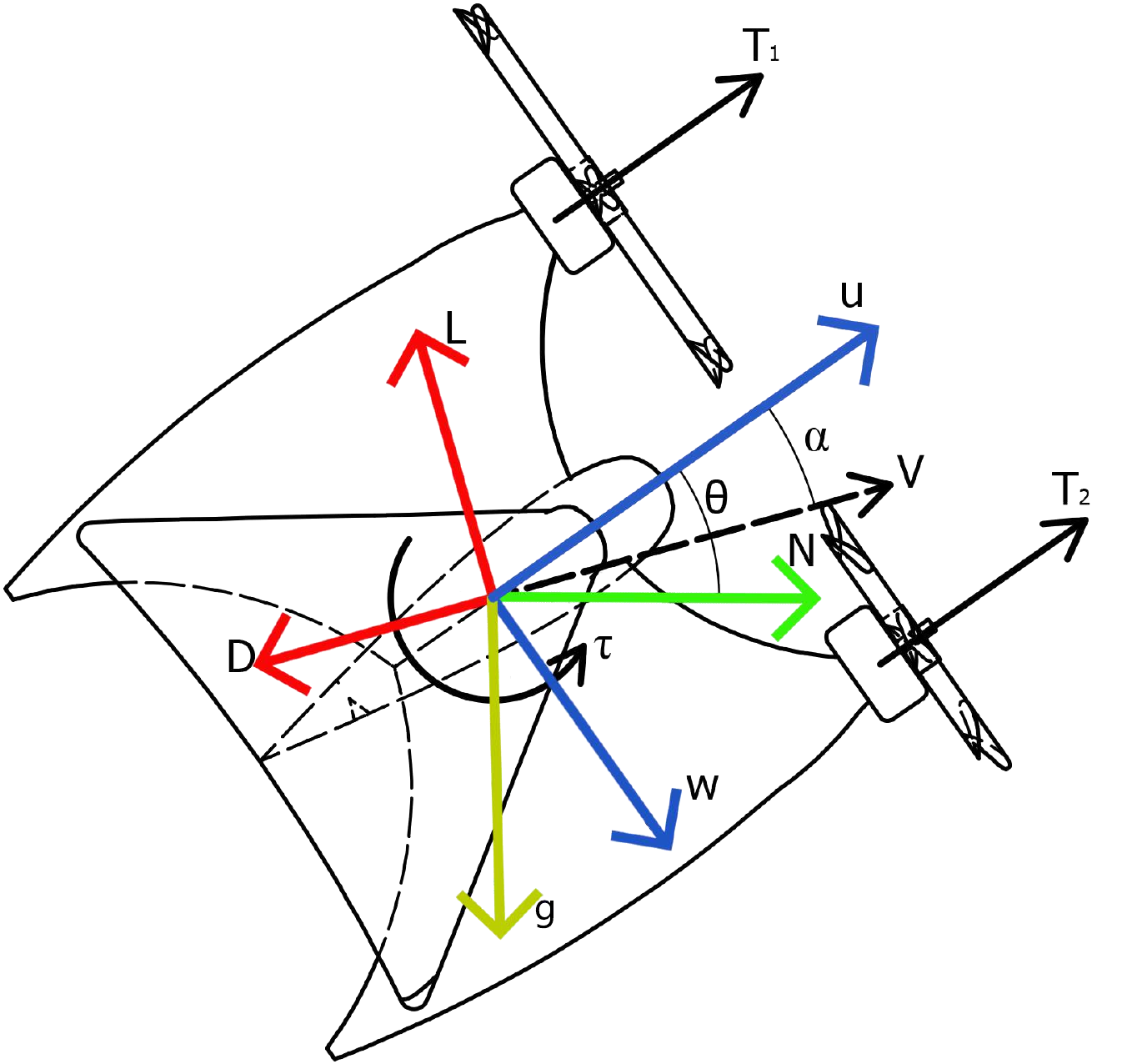}
	\caption{System forces taken into account for mathematical modeling and control design.}
	\label{forces}
\end{figure} 
\begin{eqnarray}
\label{sum1}
&\Sigma_{1}&\begin{cases}
        	\dot{u} = -D \cos\alpha + L\sin\alpha + T - g\sin\theta - qw\\
        	\dot{w} = -D \sin\alpha - L\cos\alpha + g\cos\theta + qu
	\end{cases}\\
\label{sum2}
    &\Sigma_{2}& \begin{cases}
        	\dot{\theta} = q\\
        	\dot{q} = \tau.
	\end{cases}
\end{eqnarray}
Model variables and parameters are defined in Table \ref{tab:variables}.
\begin{table}[!hbt]
\begin{center}
\begin{tabular}{|| l | l ||} 
\hline
\hline 
$u$ & $x$ body velocity expressed in body frame \\
\hline 
$w$ & $y$ body velocity expressed in body frame \\
\hline
$\theta$ & Aircraft pitch angle \\
\hline
$q$ & Aircraft angular velocity \\
$D$ & Drag force \\ 
\hline
$L$ & Lift force \\ 
\hline
$\alpha$ & Angle of Attack (AoA) \\ 
\hline
$g$ & Gravity \\ 
\hline
$T$ & Thrust force (Input control)\\
\hline
$\tau$ & Torque (Input control)\\
\hline
\end{tabular}
\caption{System model variables and parameters.}
\label{tab:variables}
\end{center} 
\end{table}
The angle of attack (AoA) is defined as
\begin{equation}
\label{eqn:alpha}
\alpha = \tan^{-1}\frac{w}{u}.
\end{equation}
Control inputs for system (\ref{sum1}-\ref{sum2}) are the thrust $T$ and the torque $\tau$ generated by the difference of thrust between rotors and the aileron angle. These terms can be described as    
\begin{eqnarray}
	\label{eqn:thrust} 
	T &=& T_{1} + T_{2} \\
    \label{eqn:tau}
    \tau &=& \tau_{\delta} + \tau_{L} + (T_{1} - T_{2})
\end{eqnarray}
where $\tau_{\delta}$ is the part of torque generated by the propellers wind stream in the aileron; $\tau_{L}$ is the torque induced by wind stream shocking in the wing area, and the rest of the torque is originated by the difference in force between the two pairs of motors. The total thrust force is calculated by the sum of each motor's thrust.

System (\ref{sum1})-(\ref{sum2}) can be divided in two subsystems, where $\sum{}{}_{1}$ and $\sum{}{}_{2}$ represent translational and attitude dynamics, respectively. From \cite{20,21}, it is known that $\sum{}{}_{2}$ is faster than $\sum{}{}_{1}$, hence $\theta$ can be used as a virtual controller for $\sum{}{}_{1}$. Subsystem $\sum{}{}_{2}$ is a double integrator and can be fully controlled by $\tau$.

\subsubsection{Aerodynamic forces}
The proposed model (\ref{sum1})-(\ref{sum2}) considers aerodynamic terms explained next. The \textit{lift} $L$ and \textit{drag} $D$ forces are nonlinear functions of several variables. The most important of these variables are: airspeed, AoA and lift and drag coefficients \cite{7}. Lift and drag are defined as follows 
\begin{eqnarray}
	\label{eqn:Lforce} 
	L &=& KC_{L}V^{2} \\
    \label{eqn:Dforce}
	D &=& KC_{D}V^{2}
\end{eqnarray}
where $K = \frac{\rho S}{2}$; $S$ is the wing hitting area and $\rho$ is air pressure; $V=\sqrt{u^{2} + w^{2}}$ is air speed magnitude. The drag and lift coefficients $(C_{L}, C_{D})$ for the airfoil used in this experiment are described at Fig. \ref{wing_forces} according to the angle of attack $\alpha$. Such coefficients correspond to a symmetrical airfoil NACA-0020 shown at Fig. \ref{airfoil}. They were obtained by numerical simulations made in the software XFR5 with many different Reynold numbers. Several aileron angles $\delta$ are considered to compute different lift and drag coefficients.
\begin{figure}[ht!]
	\centering
	\includegraphics[width=240pt, height=300pt]{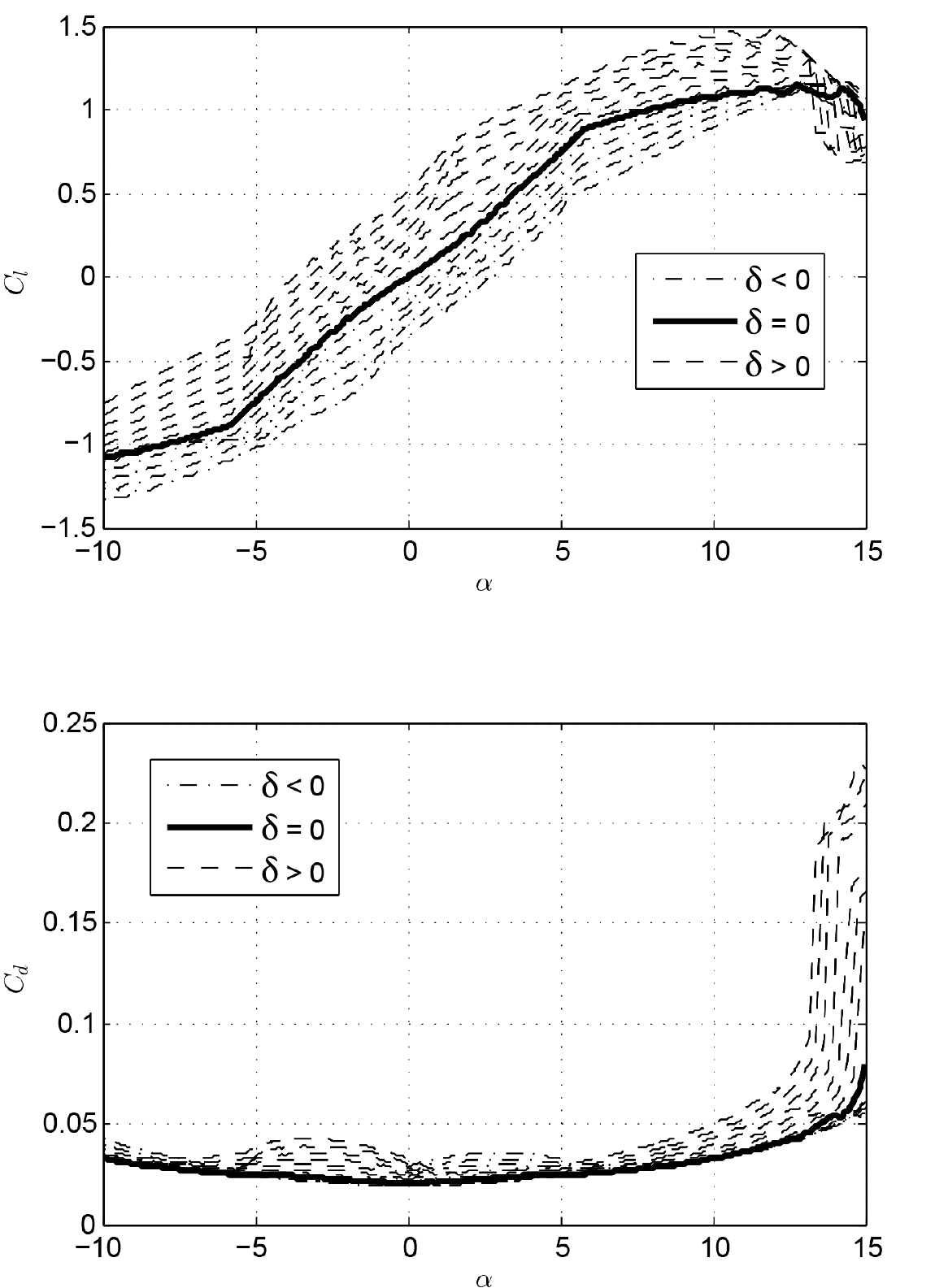}
	\caption{Lift and Drag coefficients w.r.t $\alpha$ angle obtained from XFR5 software and studies presented in \cite{18}. According to the maximum lift to drag ratio $(L/D)$ considering $C_{l}$ and $C_{d}$, it is obtained an AoA $\alpha=6$ $degrees$ tha corresponds to the optimal AoA for cruise flight mode.}
	\label{wing_forces}
\end{figure}
%
\begin{figure}[ht!]
	\centering
	\includegraphics[width=220pt]{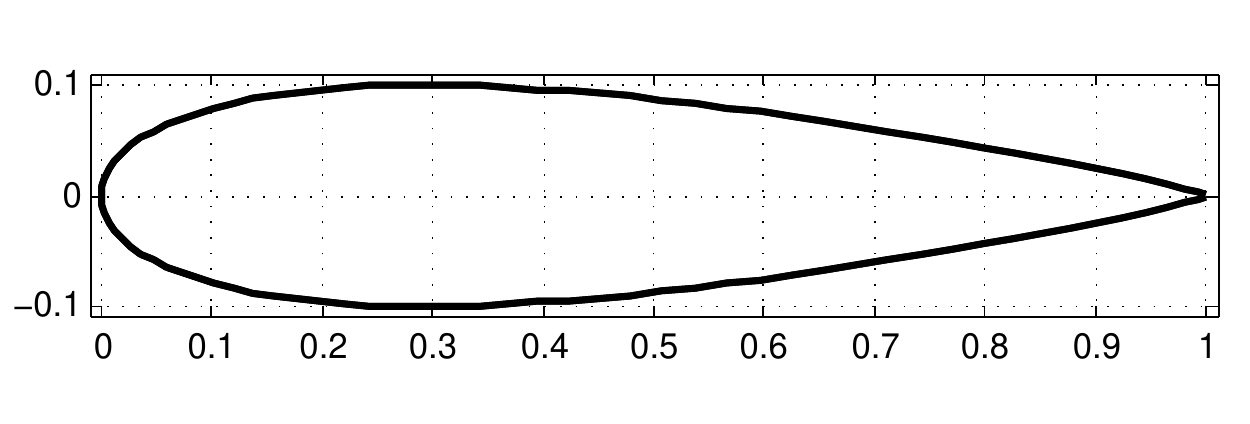}
	\caption{Normalized airfoil profile used in the aircraft (NACA-0020).}
	\label{airfoil}
\end{figure}

\subsection{Problem Statement}
The problem can be summarized in the following
\begin{problem}
Given system (\ref{sum1})-(\ref{sum2}), design a continuous and easy-to implement control with bounded inputs that stabilizes the transition maneuver from hover to cruise mode and vice-versa.
\end{problem}
%
\section{Main Result} \label{sec:result}
The key idea behind the proposed controller is explained next. Based on time-scale separation principle presented in \cite{20,21}, it is possible to use the variable $\theta$ as a virtual controller for subsystem \eqref{sum1}. Once this virtual controller is designed, its value must be tracked by the subsystem \eqref{sum2}. 

Let's start with the design of the virtual controller for $\sum_{1}$. In this case the states $u$ and $w$ should converge to predefined values according to the appropriate characteristics for the actual flight mode. Subsystem \eqref{sum1} can be rewritten as
\begin{eqnarray}
	\label{eqn:udot3}
    \dot{u} &=& f_{1}(u,w,q) + T - g\sqrt[2]{1-\epsilon^{2}}\\
    \label{eqn:wdot3}
    \dot{w} &=& f_{2}(u,w,q) + g\epsilon
\end{eqnarray}
%
where
\begin{eqnarray*}
	\epsilon &=& \cos\theta\\
    f_{1}(u,w,q) &=& -D \cos\alpha + L\sin\alpha - qw\\
	f_{2}(u,w,q) &=& -D \sin\alpha - L\cos\alpha + qu.
\end{eqnarray*}
From the previous discussion $\epsilon$ is considered as a controller, since it is a function of $theta$ angle. Before proposing the control algorithm we should define some natural assumptions and remarks:
\begin{assumption}
\label{ass:1}
In cruise flight mode it is natural to consider that $u \gg w$, i.e. the horizontal velocity is much greater than vertical velocity.
\end{assumption}
%
\begin{assumption}
\label{ass:2}
In hover mode $u\approx 0$ and $w\approx 0$.
\end{assumption}
%
\begin{remark} 
\label{rk:alpha}
From \eqref{eqn:alpha}, $u \neq 0$ .
\end{remark}
%
\begin{remark}
\label{rk:epsilon}
Since \eqref{eqn:udot3} holds, then $\mid \epsilon \mid \leq 1$.
\end{remark}
%

One important aspect of the control strategy is to design desired values of aircraft velocities $(u_{d}, w_{d})$ according to the speed characteristics of each flight mode. If the transition maneuver is from hover to cruise, $u$ must increase to a value in which the lift force is enough to maintain the aircraft at a constant altitude. On the other hand, $w$ must be designed w.r.t. $u$ in such a way that the aircraft AoA be at the optimal value considering the lift to drag ratio $(L/D)$, which for the lift and drag coefficients shown at Fig. (\ref{wing_forces}) results in $6$ degrees. Equations for $(u_{d}, w_{d})$ are computed as follows
\begin{eqnarray}
\label{eq:ud}
&u_{d}=&\begin{cases}
        	\frac{\arctan \big(a_{u}(\frac{t}{5}-L_{u})\big) }{a_{u}}+L_{u} \hspace{12pt} $if$ \hspace{10pt} \frac{t}{5} > L_{u}\\
            \frac{t}{5}  \hspace{88pt} $if$ \hspace{10pt} \frac{t}{5} \leq L_{u}
	\end{cases}\\
    \label{eq:al}
&\alpha_{d}=&\begin{cases}
        	\frac{\arctan \big(a_{\alpha}(t-L_{\alpha})\big) }{a_{\alpha}}+L_{\alpha}  & $if$ \hspace{10pt} t > L_{\alpha}\\
            t  & $if$ \hspace{10pt} t\leq L_{\alpha}
	\end{cases}\\
    \hspace{10pt}
\label{eq:wd}
    &w_{d}=& u_{d}\tan\alpha_{d}
\end{eqnarray}
where $a_{u} = \frac{\pi}{2(M_{u}-L_{u})}$, $a_{\alpha} = \frac{\pi}{2(M_{\alpha}-L_{\alpha})}$; $(M_u, L_u)$, $(M_{\alpha}, L_{\alpha})$ are positive constants which $L_u \leq M_u$, $L_{\alpha} \leq M_{\alpha}$ and $M_u$, $M_{\alpha}$ are the maximum value the functions could have and $L_u$, $L_{\alpha}$ are the linear region of the functions. For determining the desired AoA from hover to cruise mode, it is supposed that the system starts from $\alpha=0$ and ends in $\alpha=6$ degrees which can be calculated with \eqref{eq:al}. It is important to mention that $(u_d,w_d,\alpha_d)$ are continuously differentiable functions. In the same way, when the transition occurs from cruise to hover, both speeds must be reduced in such a way that the AoA decreases to the point of being equal to zero.

From Remark \ref{rk:epsilon}, the virtual control $epsilon$ is designed by using a continuous and non-decreasing saturation function $\sigma(s)$ satisfying \cite{23}:
\begin{enumerate}[(a)]
\item $s\sigma(s) > 0$	\quad\quad $\forall s \neq 0$.
\item $\mid\sigma(s)\mid \leq M_s$ \quad\space $\forall s \in \mathbb{R}$.
\item $\sigma(s) = s$ when $\mid s \mid \leq L_s$.\\
\end{enumerate}
where $L_s$ and $M_s$ are positive constants which describes the linear part of the function and the saturation value, respectively, hence $L_s \leq M_s$. It is now possible to propose a preliminary control system $(\epsilon, T)$ as follows
\begin{eqnarray}
	\label{eqn:epsilon}
    \epsilon &=& -\sigma_{2}\big(f_{2}(u,w,q) + \sigma_{1}(w - w_{d}) - \dot{w}_{d}\big)\\
    \label{eqn:T}
    T &=& -\sigma_{3}(u - u_{d}) + g\sqrt[2]{1-\epsilon^{2}} - f_{1}(u,w,q) + \dot{u}_{d}.
\end{eqnarray}
From (\ref{eqn:epsilon}) it is possible to obtain $\theta_{d}$. Then it is possible to define controller
\begin{equation}
        \label{control:tau}
      \tau = -k_{\theta}(\theta - \theta_{d}) - k_{q}(q - q_{d}) + \dot{q}_d
\end{equation}
The main result can be summarized in the following 
\begin{theorem}
Let system \eqref{sum1}-\eqref{sum2} and consider Assumptions (\ref{ass:1},\ref{ass:2}), then controllers \eqref{eqn:T}, \eqref{control:tau} with $\theta_{d}$ obtained from \eqref{eqn:epsilon} results in a local asymptotically stable closed-loop system.
\end{theorem}
\begin{proof}
Consider a coordinate change over the original system with
\begin{eqnarray*}
	x_{1} = u - u_{d} \hspace{15pt}  x_{2} = w - w_{d} \hspace{15pt} x_{3} = \theta - \theta_{d} \hspace{15pt} x_{4} = q - q_{d}.
\end{eqnarray*}
Since $x_{3}$ and $x_{4}$ states are independent of $x_{1}$ and $x_{2}$, and because it is well known that the structure of the subsystem $\sum_{2}$ is at least GAS, then the problem reduces to investigate stability of states $(x_1,x_2)$. For that, observe that the closed-loop system has a structure similar to that of \cite{23}, and then with the Lyapunov function
\begin{equation}
	\label{eqn:lyapunov2}
    V = \frac{1}{2}x_{1}^{2} + \frac{1}{2}x_{2}^{2}.
\end{equation}
and following the same procedure of \cite{23} one can obtain asymptotic stability of the closed-loop system.
\end{proof}
%
\section{Simulation Experiments} \label{sec:sims}
In this section a MATLAB simulation is performed, where two experiments are conducted: transition from hover to cruise mode and vice versa. To define the $u_{d}$ and $w_{d}$, the constants used in \eqref{eq:ud}, \eqref{eq:al} and \eqref{eq:wd} were $L_{u}=0.7$, $M_{u}=1$, $L_{\alpha}=4$, $M_{\alpha}=6$ and the saturation and gain constants for the controls were $k_{\theta}=5$, $k_{q}=3$, $L_{\sigma}=0.9$ and $M_{\sigma}=1$. Results are explained next.

\subsection{Hover to cruise flight mode}
During hover, the initial conditions were $u_{0}=0.18$, $w_{0}=0.03$ and $\theta_{0}=80$ degrees. Results are shown in Fig. \ref{fig4} where the evolution over time of $\theta$, $\alpha$ and $\tau$ are presented. Fig. \ref{fig5} depicts the thrust control as $u$ and $w$ evolve over time. According to the conditions established for the hover mode, at the end of the transition it is achieved that $u \gg w$. This is due to the conditions of cruise flight mode. Since the velocity at $u$ is much higher than $w$, it leads to a relatively small AoA as shown in the middle plot of Fig. \ref{fig4}, generating in the same way a lift force in the drone. This allows that the thrust needed to hold the aircraft in the air can be smaller as shown in the lower plot of Fig. \ref{fig5}. 

\subsection{Cruise to hover flight mode}
The simulation results for the opposite flight transition are presented with $u_{0}=0.7$, $w_{0}=0.2$ and $\theta_{0}=35$ degrees as the initial conditions. Fig. \ref{angulos2} shows the system evolution w.r.t. the desired angle and its respective AoA. In Fig. \ref{thrust2} it can be seen that the velocities decrease at $\approx0$ which corresponds to the hover condition. Since these speeds decrease, the lift and drag forces decrease in the same way as the angle of attack approaches to zero, then, to keep the drone flying, thrust force will be increased again as can be seen in the third plot of Fig. \ref{thrust2}.

\begin{figure}[ht!]
	\includegraphics[scale=0.65]{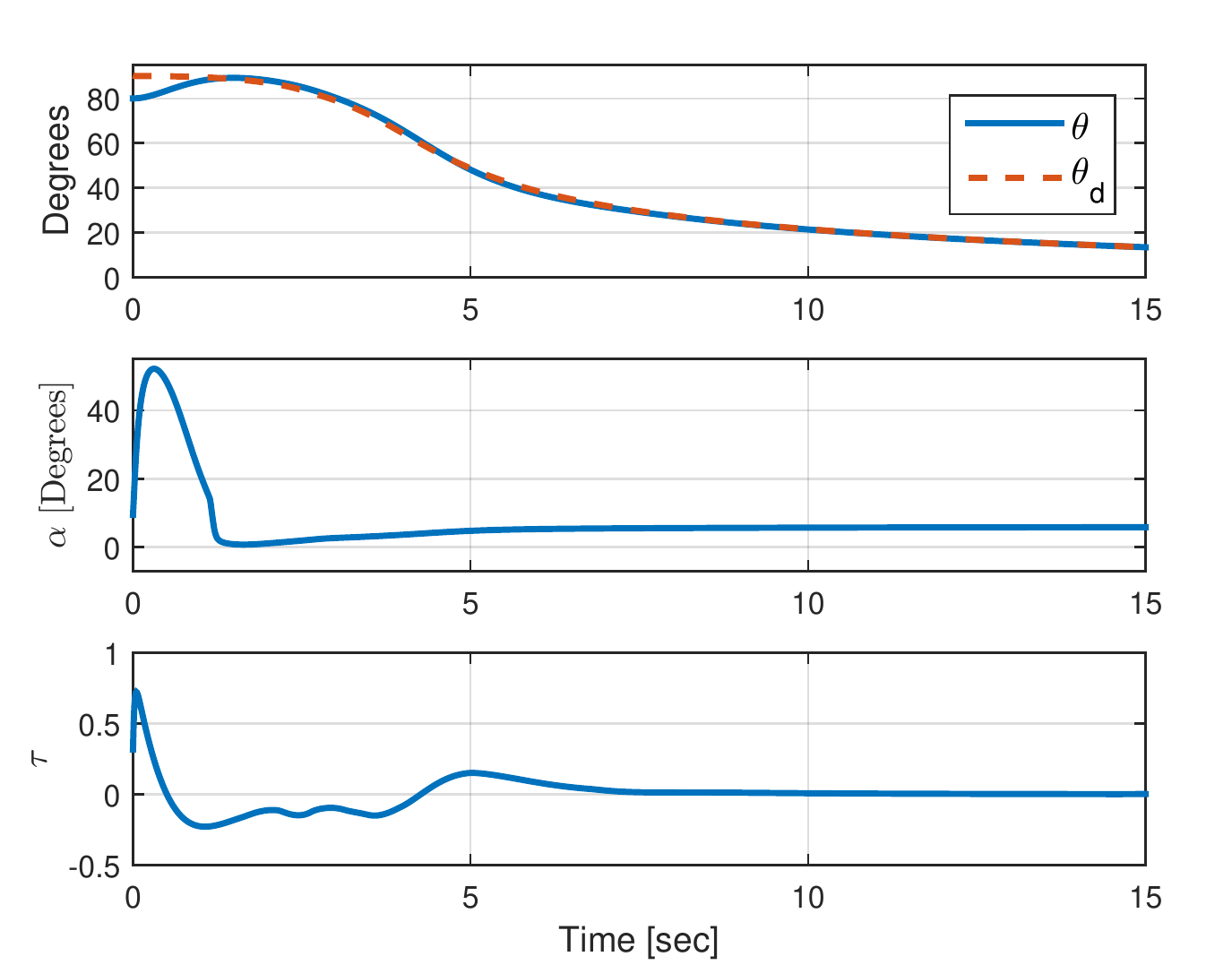}
	\centering
	\caption{$\theta$ and $\alpha$ evolution from hover to cruise flight mode according to $\theta_{d}$; also the control input $\tau$ is presented.} 
	\label{fig4}
\end{figure}
\begin{figure}[ht!]
	\includegraphics[scale=0.65]{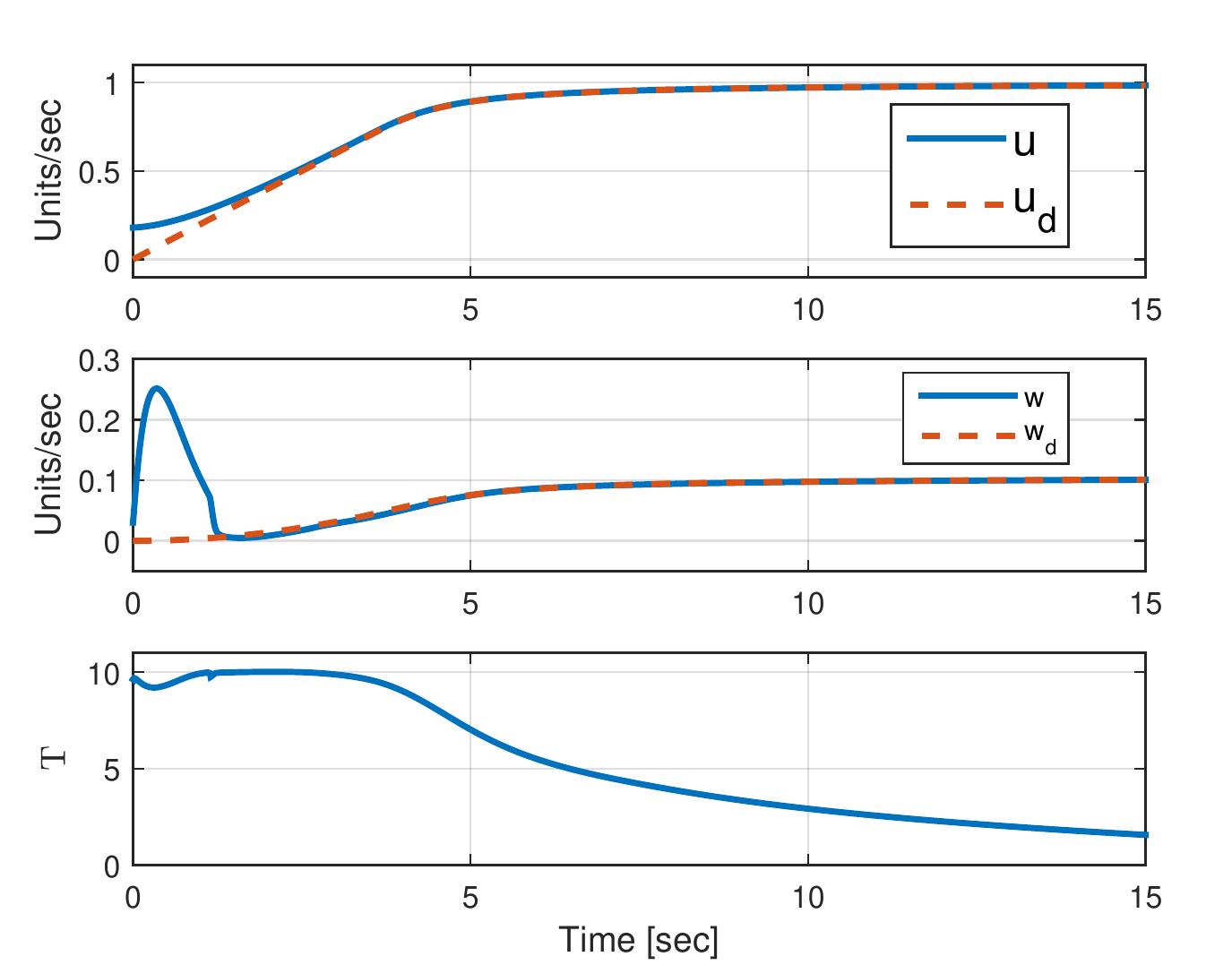}
	\centering
	\caption{Velocities $u$ and $w$ during the transition from hover to cruise flight mode, also the thrust control is presented.}
	\label{fig5}
\end{figure}
\begin{figure}[ht!]
	\includegraphics[scale=0.65]{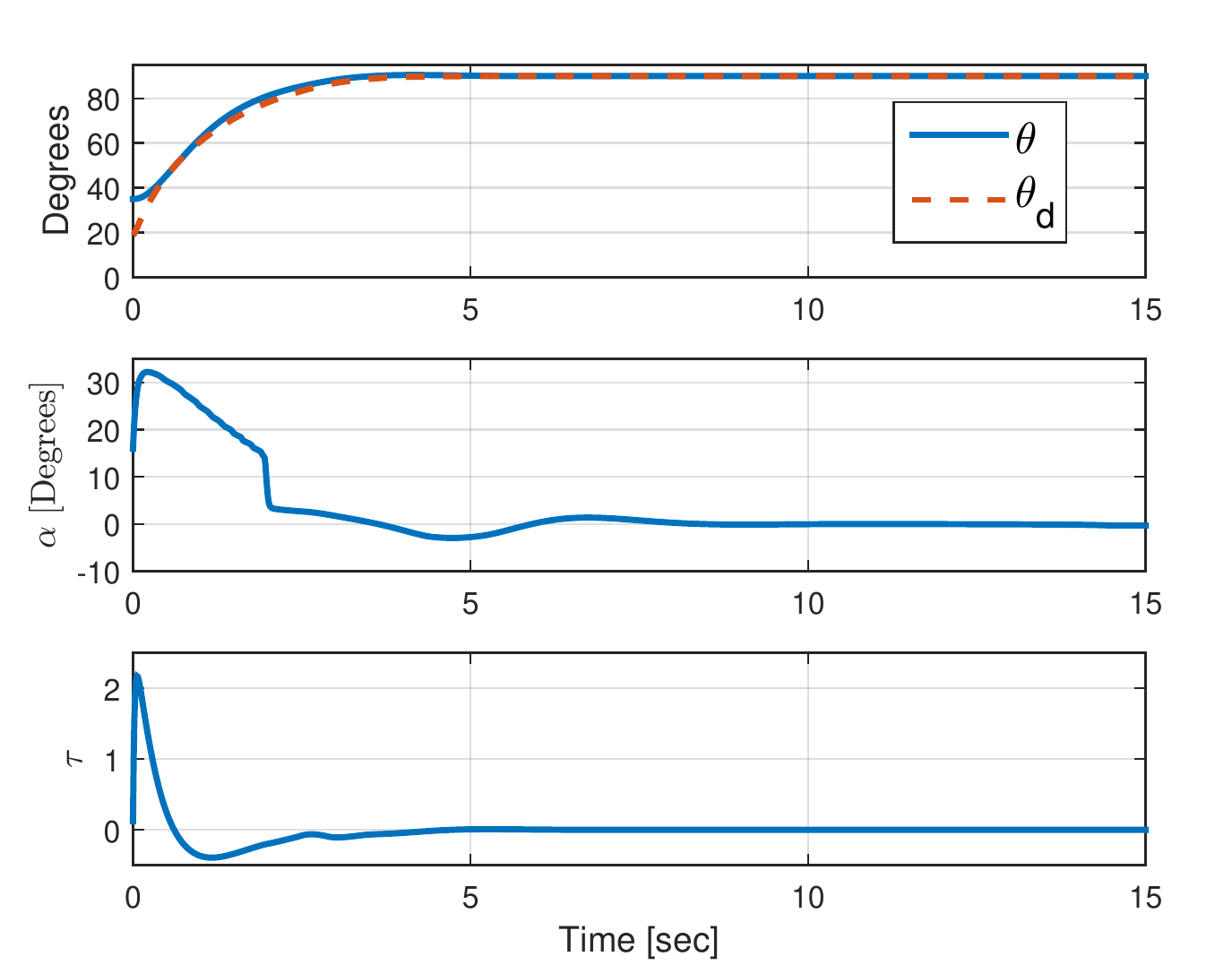}
	\centering
	\caption{$\theta$ and $\alpha$ angle evolution from cruise to hover flight mode according to $\theta_{d}$; also the control input $\tau$ is presented.} 
	\label{angulos2}
\end{figure}
\begin{figure}[ht!]
	\includegraphics[scale=0.65]{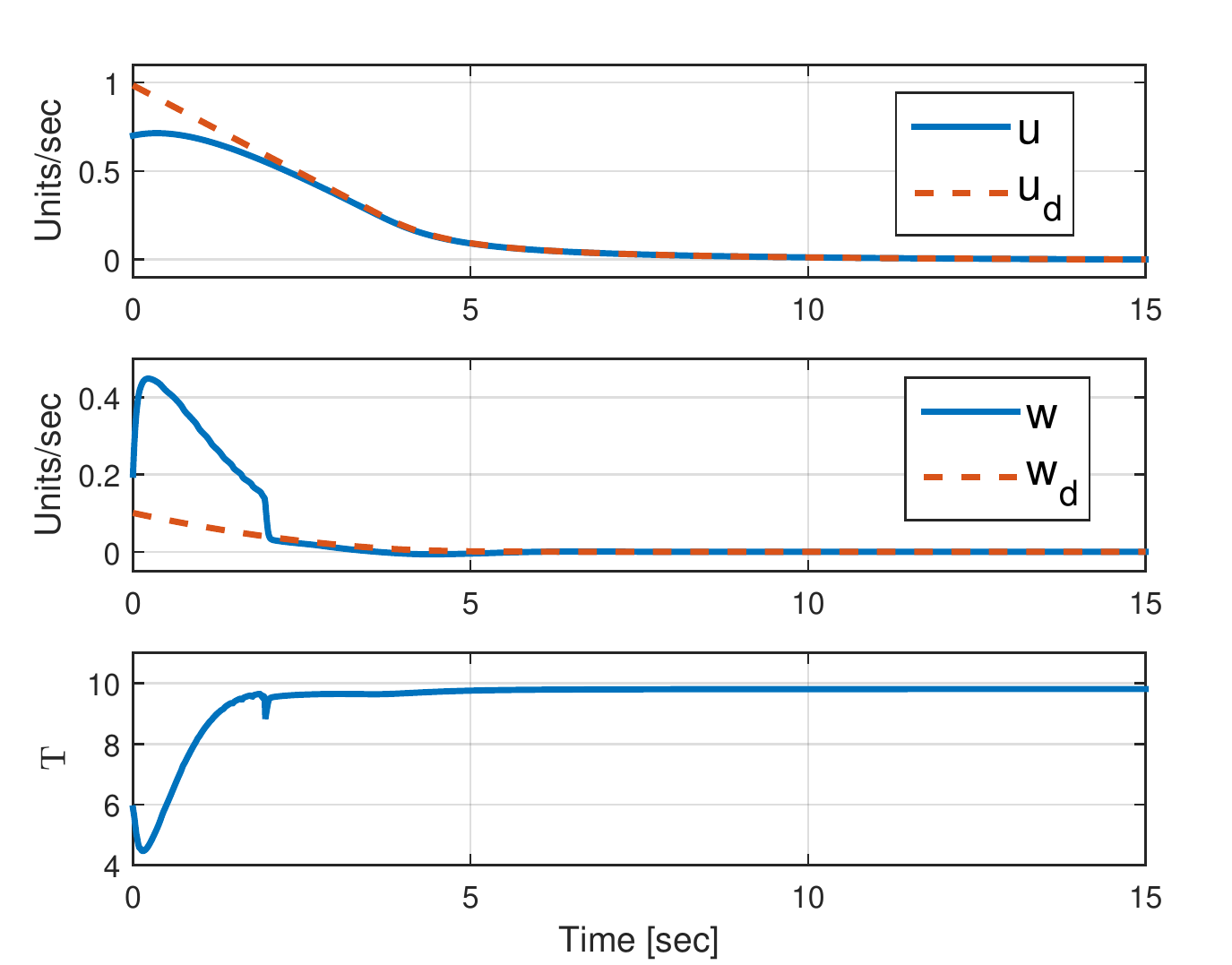}
	\centering
	\caption{$u$ and $w$ velocities during the transition from cruise to hover flight mode, with the respective thrust control.} 
	\label{thrust2}
\end{figure}
\section{Conclusions} \label{sec:conclusions}
In this paper, a simple control strategy for the transition maneuver of the tail-sitter UAV is proposed. Such controller is based on the time-scale separation and saturation functions. The design is based on Lyapunov approach. Simulations demonstrate the effectiveness of the controller for achieving transition from hover to cruise mode and vice versa. It is important to mention that the controller have the peculiarity that it does not present any switching, it is smooth and it takes into account the saturation limits imposed by the actuators. Such characteristics are useful for implementation in a real UAV.
    
In a further work, we can execute a real experiment with the physical prototype of the UAV and investigate the behavior of the real system. Also it could be possible to design a robust control that consider external disturbance as part of the modeling, such as wind gusts terms. Also the control can be designed by analyzing the 6-DOF model.

\bibliographystyle{IEEEtran}
\bibliography{biblio}
\end{document}